\documentclass[default,iicol]{sn-jnl}

\usepackage{enumitem}
\usepackage{booktabs}
\usepackage{subcaption}

\jyear{2021}%

\theoremstyle{thmstyleone}%
\newtheorem{theorem}{Theorem}
\theoremstyle{thmstyletwo}%

\theoremstyle{thmstylethree}%
\newtheorem{definition}{Definition}%

\usepackage{bm}

\def\eqref#1{equation~\ref{#1}}

\def\1{\bm{1}}

\def\vp{{\bm{p}}}
\def\vq{{\bm{q}}}

\DeclareMathAlphabet{\mathsfit}{\encodingdefault}{\sfdefault}{m}{sl}
\SetMathAlphabet{\mathsfit}{bold}{\encodingdefault}{\sfdefault}{bx}{n}

\begin{document}

\title[Article Title]{Signal Observation Models and Historical Information Integration in Poker Hand Abstraction}


\author[1,2]{\fnm{Yanchang} \sur{Fu}}\email{fuyanchang2020@ia.ac.cn}
\author[2]{\fnm{Pei} \sur{Xu}}\email{xupei2018@ia.ac.cn}
\author[3]{\fnm{Dongdong} \sur{Bai}}\email{baidongdong@nudt.edu.cn}
\author[1,2]{\fnm{Lingyun} \sur{Zhao}}\email{zhaolingyun2021@ia.ac.cn}
\author*[2]{\fnm{Kaiqi} \sur{Huang}}\email{kqhuang@nlpr.ia.ac.cn}

\affil*[1]{\orgdiv{School of Artificial Intelligence}, \orgname{University of Chinese Academy of Sciences}}

\affil[2]{\orgdiv{Center for Research on Intelligent System and Engineering}, \orgname{Institute of Automation, Chinese Academy of Sciences}}

\affil[3]{\orgname{China RongTong Artificial Intelligence Research Center}}







\abstract{Hand abstraction has been instrumental in developing powerful AI for Texas Hold'em poker, a widely studied testbed for imperfect information games (IIGs). Despite its success, the hand abstraction task lacks robust theoretical tools, limiting both algorithmic innovation and theoretical progress. To address this, we extend the IIG framework with the \textbf{signal observation ordered game} model and introduce \textbf{signal observation abstraction} to formalize the hand abstraction task. We further propose a novel evaluation metric, the \textbf{resolution bound}, to assess the performance of signal observation abstraction algorithms. Using this metric, we uncover critical limitations in current state-of-the-art algorithms, particularly the significant information loss caused by the enforced omission of historical information. To resolve these issues, we present the \textbf{KrwEmd} algorithm, which effectively incorporates historical information into the abstraction process. Experiments in the Numeral211 hold'em environment demonstrate that KrwEmd addresses these limitations and significantly outperforms existing algorithms.}




\maketitle

\section{Introduction} \label{sec:intro}

In recent years, AI in hold'em games has achieved groundbreaking milestones, marked by the success of agents such as DeepStack, Libratus, and Pluribus~\citep{moravvcik2017deepstack, brown2018superhuman, brown2019superhuman}. Each of these AI systems has managed to defeat top human professionals. These successes demonstrate the potential of artificial intelligence in mastering complex strategic reasoning and establish hold'em games as a premier testbed for research in \textbf{imperfect information games (IIGs)}~\citep{sandholm2010state}.

Hold'em games, particularly \textbf{heads-up no-limit hold’em (HUNL)}, have long been recognized as a challenging benchmark for IIGs. Unlike \textbf{perfect information games (PIGs)} such as chess or Go~\citep{campbell2002deep, silver2016mastering}, hold'em games require players to make decisions under uncertainty due to incomplete knowledge of their opponents' private information. This challenge of reasoning under hidden information and stochasticity mirrors core problems in real-world decision-making scenarios.

A class of highly effective methods for addressing IIGs is based on the \textbf{counterfactual regret minimization (CFR)} algorithm and its numerous extensions~\citep{zinkevich2007regret, lanctot2009monte, tammelin2014solving, brown2019solving}. CFR iteratively adjusts strategies to minimize regret, eventually converging toward an approximate Nash equilibrium. However, while CFR and its variants are computationally feasible for smaller-scale games, they face substantial challenges when applied to human-scale hold'em games. To provide perspective, using CFR to solve the full \textbf{heads-up limit hold’em (HULH)} variant, a simplified two-player variant of HUNL, requires handling an enormous state space exceeding 1.2 petabyte (PB) of memory~\cite{johanson2013measuring}. For larger, more complex hold'em variants, direct CFR computation is infeasible due to constraints on both computational power and storage capacity.

To address the prohibitive computational demands of large-scale hold'em games, the research community has developed the \textbf{abstraction-solving-translation} paradigm, which has become a cornerstone of modern hold'em AI. This paradigm involves three distinct phases:

\begin{enumerate}[label=\textbf{\arabic*.}, leftmargin=*] 
    \item \textbf{Abstraction}: Reducing the complexity of the game by grouping similar actions or states into abstract representations. Abstractions reduce the size of the game tree, making it more computationally tractable.
    
    \item \textbf{Solving}: Applying algorithms like CFR to compute an optimal strategy for the abstracted game, which is computationally manageable.
    
    \item \textbf{Translation}: Mapping the strategy from the abstracted game back to the original, larger game, allows the AI to play the full game effectively.
\end{enumerate}

Within this paradigm, research has primarily focused on two major types of abstraction. The first is \textbf{action abstraction}, which simplifies the action space, such as bet sizes, by clustering possible actions into discrete categories. The second type is \textbf{hand abstraction}, which involves grouping similar private hand holdings into equivalence classes. For instance, hands with similar win probabilities or comparable structures, such as similar card rankings and suits, might be clustered into the same abstracted category.


    

This paper focuses on the modeling and algorithmic optimization of the hand abstraction task. Current research on hand abstraction lacks theoretical analysis tools, primarily because hold'em games are modeled as IIGs, where hand abstraction falls outside the scope of the framework—unlike action abstraction, which is a well-defined concept in IIGs. Specifically, in HULH, action abstraction groups action sequences (e.g., \texttt{/dbc/db} for deal, bet, call, deal, bet, and \texttt{/dbc/dk} for deal, bet, call, deal, check) into equivalent categories, each representable as an action sequence in IIGs. In contrast, hand abstraction groups hands (e.g., \texttt{[A$\heartsuit$K$\heartsuit \mid$ 1$\heartsuit$2$\heartsuit$3$\heartsuit$]} and \texttt{[A$\heartsuit$Q$\heartsuit \mid$ 1$\heartsuit$2$\heartsuit$3$\heartsuit$]}) into a single category. These hands cannot be directly modeled in IIGs, where the closest concept to a hand is an infoset (e.g., information set). Yet, a hand alone does not qualify as an infoset. Rather, an infoset is represented by a tuple, such as \texttt{([A$\heartsuit$K$\heartsuit \mid$ 1$\heartsuit$2$\heartsuit$3$\heartsuit$], /dbc/db)}.

The concept of \textbf{games with ordered signals}~\citep{gilpin2007lossless}, a subset of IIGs, introduces a formal framework for hand abstraction to construct a lossless abstraction based on hand isomorphisms. However, this approach relies on a complex forest structure and is confined to card-drawing models without replacement, limiting its generalizability to broader classes of games. Furthermore, this framework was primarily designed for constructing the \textbf{lossless isomorphism (LI)} algorithm, which suffers from extremely low compression rates. Since LI identifies hand equivalence classes by analyzing the structure of hands themselves, it cannot quantify differences between hands, preventing further clustering. This limitation makes it less suitable for constructing high-level AI, which has led to a gradual shift away from both the LI algorithm and the games with ordered signals framework. In practice, \textbf{outcome-based abstraction}, a lossy abstraction methodology, has gained prominence as a more effective alternative. This approach determines abstraction through a bottom-up analysis based on the results of game rollouts, enabling the quantification of differences between hands. Existing outcome-based hand abstraction methods, such as \textbf{expectation hand strength (EHS)}\citep{gilpin2007better} and \textbf{potential-aware abstraction algorithms (PAAs)}\citep{gilpin2007potential}, have significantly advanced practical abstraction techniques and remain widely used in constructing high-level AI. However, these methods rely on indirect evaluation of abstraction quality, measuring performance through strategies derived from game-solving solutions. They lack guidance from formal mathematical models, such as games with ordered signals, to design more effective and principled algorithms.

The main contributions of this paper are as follows:
\begin{itemize}[label=\textbullet, leftmargin=*]
    \item We propose the \textbf{signal observation ordered game (SOOG)} model to represent hold'em games, enabling the use of signal observation abstraction to model the hand abstraction task (Section~\ref{sec:mathematical-model}).
    \item We introduce a novel metric—\textbf{resolution bound}—to evaluate signal observation abstraction algorithms, providing a systematic framework for their design and evaluation (Section~\ref{sec:resolution-bound}).
    \item Using resolution bound, we identify deficiencies in the resolution bound of current state-of-the-art (SOTA) algorithms, revealing that enforced forgetting of historical information leads to significant information loss (Section~\ref{sec:PAAs-and-excessive-abstraction}).
    \item To address these deficiencies, we propose the \textbf{KrwEmd algorithm}, which achieves a higher resolution bound by effectively leveraging historical information (Section~\ref{sec:k-recall-methods}).
\end{itemize}

We conduct extensive experiments in the Numeral211 hold'em environment, demonstrating the superiority of the proposed KrwEmd algorithm over existing SOTA algorithms in terms of resolution bound and overall performance (Sections~\ref{sec:experimental-setup} and~\ref{sec:experiment}). To provide a comprehensive understanding of the study's context, we introduce the rules and structure of hold'em games in Section~\ref{sec:holdem}.

\section{Hold'em Games}\label{sec:holdem}

Hold'em games are a class of poker games characterized by the use of private (hole) cards combined with community (board) cards to construct the best possible n-card poker hand. Players make decisions across multiple betting phases, where each player, based on their current hand and observations of other players’ actions (which may reflect the strength of their hole cards), responds with actions such as betting, calling, raising, or folding. The outcome is determined either when only one player remains or by a final showdown where players reveal their private cards, and the player with the best hand wins the pot.


Hold’em has several well-known variants, each differing in betting structure, hand card count, and gameplay complexity. Below, we highlight key rules of HUNL and HULH, followed by the introduction of Numeral211 hold'em, a specialized environment proposed to facilitate hand abstraction research.

\paragraph{HULH and HUNL} 
HULH and HUNL are two of the most prominent benchmark environments for AI research in IIGs. Both games feature two players who receive two private hole cards and compete using five community cards revealed in four betting phases: Preflop, Flop, Turn, and River. Players make decisions at each round, with possible actions including betting, calling, raising, checking, and folding. The key difference lies in the betting structure: in HULH, bets are of fixed size at each phase, whereas in HUNL, players can bet any amount up to their entire stack. As a result, HUNL has a larger action space, leading to greater strategy complexity, while HULH’s fixed bets allow for more tractable analysis and solution methods.

\paragraph{Numeral211 Hold'em} 
Numeral211 hold'em is a novel environment proposed in this paper specifically for hand abstraction research. Its design compresses the rules of hold’em games into a smaller framework that retains the essential characteristics of hole cards while ensuring that strategies can be evaluated on standard computational hardware. This enables theoretical performance assessments while maintaining key strategic elements of larger games like HULH.

Numeral211 hold'em is a two-player game where each player antes 5 chips into the pot before the start of a game. The game is played with a reduced deck of 40 cards, comprising four suits (spades, hearts, clubs, diamonds). Each suit contains ten numeric cards: 2 through 9, plus the Ace. The game proceeds through three betting phases. The first phase deals two private cards to each player, and the next two phases reveal one community card each. The player who acts first in phase 1 is player 1, while in the second and third phases, player 2 acts first. During each phase, the first player to act may either bet a fixed number of chips (10 in phase 1 and 20 in phases 2 and 3) to start a betting phase (bet) or pass this right to the opponent (check). Once a betting phase begins, players take turns to raise (10 in phase 1 and 20 in phases 2 and 3, with a combined maximum of 3 raises per phase for both players), call, or fold. If both players check or one player calls the other's action, the game progresses to the next phase or concludes. Folding means forfeiting the chance to win the hand and awarding the entire pot to the opponent. Only if the game reaches the final phase and does not end in a fold will a showdown occur, where the players reveal their hole cards and the pot is awarded according to the relative strength of their hands, specifically the strongest three-card combination.
Only if the game reaches the final phase and does not end in a fold will a showdown occur, where the players reveal their hole cards, and the pot is awarded based on the relative strength of their hands, specifically the strongest three-card combination, as summarized by the rank of hands illustrated in Figure~\ref{fig:handranks}.

\begin{figure}[h]%
\centering
\includegraphics[width=0.5\textwidth]{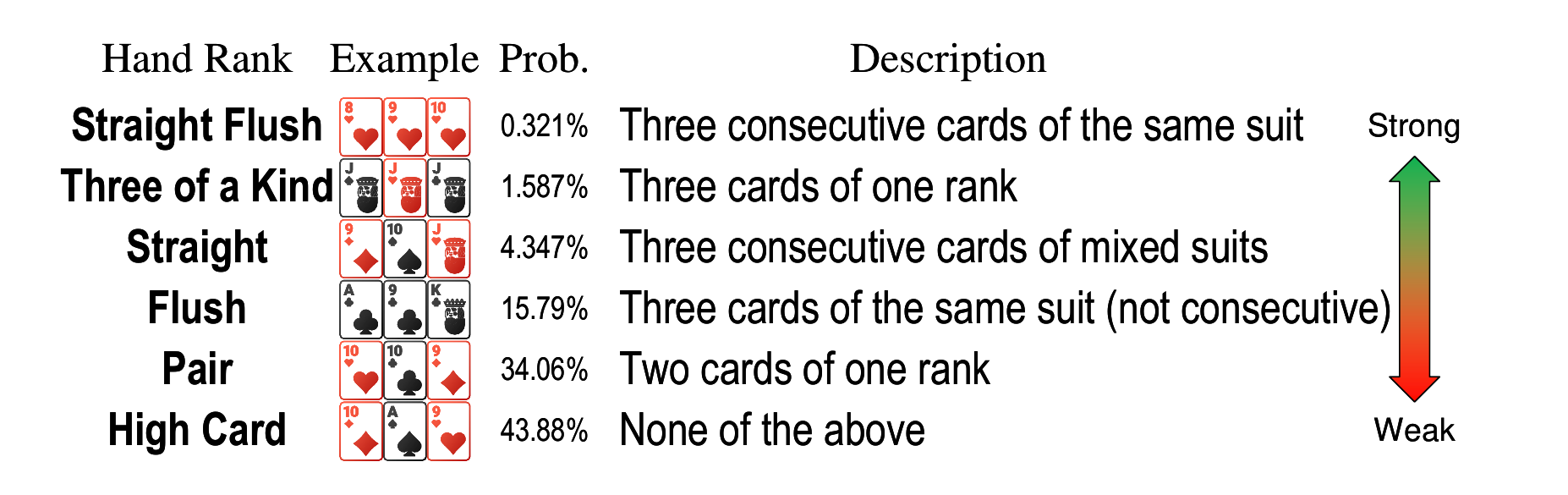}
\caption{The hand ranks of Numeral211 hold'em.}\label{fig:handranks}
\end{figure}

\section{Mathematical Modeling}\label{sec:mathematical-model}

In this section, we introduce the foundational concepts necessary for signal observation in Hold'em games. We begin by defining \textbf{signal observation ordered games}, a specialized game framework that captures the structure and dynamics of Hold'em games (Section \ref{subsec:signal-ordered-games}). Building on this framework, we establish the concept of \textbf{signal observation abstraction}, which models the hand abstraction task (Section \ref{subsec:signal-observation-abstraction}).

\subsection{Signal Observation Ordered Games} \label{subsec:signal-ordered-games}

\begin{figure*}[t]
\centering 
\begin{align}
&\pi_i(\sigma, (\tilde{v}, \theta))= 
\begin{cases}
    1 & \text{if }(\tilde{v}, \theta)=(\tilde{v}_0, \theta_0), \\
    \pi_i(\sigma, (Pa_{\tilde{V}}(\tilde{v}), Pa_{\Theta}(\theta))) & \text{if }Pa_{\tilde{V}}(\tilde{v}) \in \tilde{H}^+ \backslash \tilde{H}_i, \\ 
    \sigma(Pa_{\tilde{V}}(\tilde{v}), \vartheta_i(Pa_{\Theta}(\theta)), Pa_{\tilde{A}}(\tilde{v})) \pi_i(\sigma, (Pa_{\tilde{V}}(\tilde{v}), Pa_{\Theta}(\theta))) & \text{if }Pa_{\tilde{V}}(\tilde{v}) \in \tilde{H}_i.
\end{cases} \label{eq:reach-i}
\end{align} 
\end{figure*}

\begin{definition}
The structured  tuple $\tilde{\mathcal{G}} = \left\langle \tilde{\mathcal{T}}, \tilde{N}, \rho, \tilde{A}, Pa_{\tilde{A}}, \gamma, \varTheta, \varsigma, \vartheta, \omega, \succeq, u \right\rangle$ formally defines a signal observation ordered game (SOOG), where:
\begin{itemize}[left=0cm]
    \item $\tilde{\mathcal{T}} = \left\langle \tilde{V}, \tilde{v}^0, \tilde{Z}, Pa_{\tilde{V}} \right\rangle$ is a public tree consisting of a finite set of public nodes $\tilde{V}$, a unique initial node $\tilde{v}^0 \in \tilde{V}$, a finite set of terminal public nodes $\tilde{Z} \subseteq \tilde{V}$, and a predecessor function $Pa_{\tilde{V}}: \tilde{V}^+ \to \tilde{H}$, mapping each non-initial node $\tilde{v}^+ \in \tilde{V}^+$ to its immediate predecessor $\tilde{h} \in \tilde{H}$. Here, $\tilde{V}^+ = \tilde{V} \backslash \{\tilde{v}^0\}$ is the set of non-initial nodes, while $\tilde{H} = \tilde{V} \backslash \tilde{Z}$ is the set of internal nodes.
    
    \item $\tilde{N} = N \cup \{sp\}$ is a finite set of augmented players, where $sp$ refers to a spectator who observes the public information of the game without influencing its progression. $N = \{0, 1, \dots, n\}$ denotes the set of players, with $0$ representing a special player, commonly referred to as \textbf{chance} or \textbf{nature}, whose actions correspond to random events. The set of rational players (i.e., non-chance players) is denoted by $N_+ = N \backslash \{0\}$, and the set of augmented rational players is given by $\tilde{N}_+ = N_+ \cup \{sp\}$. The player function $\tilde{\rho}: \tilde{H} \to N$ partitions the set $\tilde{H}$ among players. The decision public node set is defined as $\tilde{H}_+ = \bigcup_{i \in N_+} \tilde{H}_i$, where $ \tilde{H}_i = \{\tilde{h} \in \tilde{H} \mid \rho(\tilde{h}) = i\}$.
    
    \item $\tilde{A} = A_+ \cup \tilde{A}_0$ is a finite set of actions. The set $A_+$ includes the actions available to rational players, while $\tilde{A}_0 = \{a_0\}$ denotes the set of actions available to the chance player. Notably, $\tilde{A}_0$ contains only one action, $a_0$, which represents a placeholder action in the public tree where the chance player reveals a signal. The function $Pa_{\tilde{A}}: \tilde{V}^+ \to \tilde{A}$ defines an action partition of $\tilde{V}^+$, mapping each non-initial public node $\tilde{v}^+$ to the action $a \in \tilde{A}$ that immediately leads to the occurrence of $\tilde{v}^+$. The function $\mathcal{A}(\tilde{h} \in \tilde{H}) = \{a^\prime \in \tilde{A} \mid [\exists \tilde{v}^+ \in \tilde{V}^+] (Pa_{\tilde{A}}(\tilde{v}^+) = a^\prime \wedge Pa_{\tilde{V}}(\tilde{v}^+) = \tilde{h})\}$ confines the available actions of each internal public node.
    
    \item $\gamma: \tilde{V} \to \mathbb{N}^+$ is a phase partition of $\tilde{V}$, assigning to each public node $\tilde{v}$ a value corresponding to the number of chance public nodes encountered along the path from the initial public node $\tilde{v}^0$ to and including $\tilde{v}$, thereby defining the phase of $\tilde{v}$. The maximum phase in the game is denoted by $\Gamma$, and notably, $\gamma(\tilde{v}^0) = 1$, indicating that $\tilde{v}^0$ is a chance public node.
    
    \item $\varTheta = \left\langle \Theta, \theta^0, Pa_{\Theta} \right\rangle$ is a signal tree with height $\Gamma$ ($\Gamma+1$ layers), consisting of a finite set of signals $\Theta$, a unique initial signal $\theta^0$, and a predecessor function $Pa_{\Theta}: \Theta^{+} \to \Theta \setminus \Theta^{(\Gamma)}$, mapping each non-initial signal to its immediate predecessor. Here, $\Theta^{(r)}$ denotes the set of signals revealed in phase $r = 1, \dots, \Gamma$; specifically, $\Theta^{(0)} = \{\theta^0\}$. $\Theta^{+} = \Theta \backslash \Theta^{(0)}$ represents the set of non-initial signals. The depth of a signal $\theta \in \Theta$ is denoted by $d_\Theta(\theta)$, and all terminal signals in $\Theta$ (i.e., signals without successors) necessarily have a depth of $\Gamma$. 
    
    \item $\varsigma: \Omega \mapsto [0,1]$ is a chance probability function that assigns a probability of occurrence to each successive pair of signals, with $\Omega=\{(\theta, \theta^\prime) \in \Theta \times (\Theta \setminus \Theta^{(\Gamma)}) \mid Pa_{\Theta}(\theta^\prime) = \theta\}$. Additionally, for each $\theta \in \Theta \setminus \Theta^{(\Gamma)}$, the sum of the probabilities for all $\theta^\prime \in S(\theta)$ is equal to 1, where $S(\theta) = \{\theta^\prime \in \Theta^{+} \mid Pa_{\Theta}(\theta^\prime) = \theta\}$.
    
    \item $\vartheta = (\vartheta_1, \dots, \vartheta_n, \vartheta_{sp})$ is a tuple of observation functions, with $\vartheta_i: \Theta \mapsto \Psi_i$ mapping each $\theta \in \Theta$ to its corresponding signal observation infoset, such that signals within the same signal observation infoset $\psi \in \Psi_i$ cannot be distinguished by the augmented rational player $i \in \tilde{N}_+$. Furthermore, all elements $\psi \in \Psi_i$ collectively form a partition of $\Theta$.
    
    \item $\omega = (\omega_1, \dots, \omega_n)$ is a tuple of survival functions, where $\omega_i(\tilde{v} \in \tilde{V}) := \1_{\text{player } i \text{ is still participating at } \tilde{v}}$.

    \item There is a total order $\succeq(\theta \in \Theta^{(\Gamma)}, i \in N_+, j \in N_+) := \allowbreak \1_{\text{player } i \text{ is ranked no lower than player } j \text{ at }\theta}$ order over the terminal signals with respect to the set of players $N_+$.
    
    \item Signals and public nodes constitute the nodes of an ordered game. The corresponding sets are defined as follows:
    \begin{itemize}[left=0cm]
        \item $H_0^{(r)} = \tilde{H}_0^{(r)} \times \Theta^{(r-1)}$, and for $j \in N_+$, $H_j^{(r)} = \tilde{H}_j^{(r)} \times \Theta^{(r)}$, $r = 1, \dots, \Gamma$, where $\tilde{H}_i^{(r)} = \{\tilde{h}\in \tilde{H}_i \mid \gamma(\tilde{h}) = r\}$, $i\in N$.
        \item $Z^{(r)} = \tilde{Z}^{(r)} \times \Theta^{(r)}$, where $\tilde{Z}^{(r)} = \{\tilde{z}\in \tilde{Z} \mid \gamma(\tilde{h}) = r\}$, $r = 1, \dots, \Gamma$.
        \item $H^{(r)} = \bigcup_{i\in N} H_i^{(r)}$ and $V^{(r)} = Z^{(r)} \cup H^{(r)}$.
        \item $H = \bigcup_{r = 1}^\Gamma H^{(r)}$, $Z = \bigcup_{r = 1}^\Gamma Z^{(r)}$, and $V = \bigcup_{r = 1}^\Gamma V^{(r)}$.
    \end{itemize}
    
    \item $u = (u_1, \dots, u_n)$ is a tuple of utility functions, where each $u_i: Z \mapsto \mathbb{R}$ represents the utility function specifically evaluating the payoff of rational player \( i \). In the final phase, for $z = (\tilde{z}, \theta) \in Z^{(\Gamma)}$, it is required that if $\omega_i(\tilde{z}) \omega_j(\tilde{z}) \succeq (\theta, i, j) = 1$, then $u_i(\tilde{z}, \theta) \geq u_j(\tilde{z}, \theta)$.
\end{itemize}
\end{definition}

Rational players make decisions based on their signal observation infosets and the current non-terminal public node. Signals within the same signal observation infoset necessarily share the same depth. Accordingly, the notation \(d_\Theta(\psi)\), originally defined for individual signals, is naturally extended to denote the depth of the entire signal observation infoset \(\psi \in \Psi_i\). A rational player has access to more information than the spectator, including all information available to the spectator. For any $i \in N_+$ and $\forall \theta \in \Theta$, we have $\vartheta_i(\theta) \subseteq \vartheta_{sp}(\theta)$.
SOOGs are a subset of IIGs, as they extend the structure of IIGs by refining nodes into two orthogonal components: public nodes and signals. Correspondingly, infosets in IIGs are decomposed into two orthogonal parts: public nodes and signal observation infosets. This decomposition allows us to isolate and study the signal observation infosets independently, serving as a prototype for hands in hold'em games. Furthermore, the strategic framework and Nash equilibrium concepts from IIGs can be naturally applied to SOOGs within this refined structure.

A rational player $i \in N_+$ chooses a strategy $\sigma_i: Q_i \mapsto [0,1]$ from $\Sigma_i$, the set of available strategies for player $i$. Here, $Q_i = \{(\tilde{h}, \psi, a) \in \tilde{H}_i \times \Psi_i \times A_+ \mid \gamma(\tilde{h}) = d_\Theta(\psi) \wedge a \in \mathcal{A}(\tilde{h})\}$, and the condition $\sum_{a \in \mathcal{A}(\tilde{h})} \sigma_i(\tilde{h}, \psi, a) = 1$ must be satisfied. When all rational players select their strategies, a strategy profile $\sigma: Q \mapsto [0,1]$ is formed, where $\sigma = \oplus_{i\in N_+} \sigma_i \in \Sigma$ \footnote{Given the functions$f_1: A_1 \mapsto B_1$ and $f_2: A_2 \mapsto B_2$, a new function $f = f_1\oplus f_2$ is defined such that $f: A_1 \oplus A_2 \mapsto B_1 \cup B_2 $, with $$f(x) = \begin{cases} f_1(x) & \text{if }x\in A_1\setminus A_2, \\ f_2(x) & \text{if }x \in A_2 \setminus A_1.\end{cases}$$} and $Q = \bigcup_{i\in N_+} Q_i$, and the probability of reaching each node $v = (\tilde{v}, \theta) \in V$ can be recursively computed using equations \ref{eq:reach-i}, \ref{eq:reach-theta} and \ref{eq:reach}, where $\tilde{H}^+ = \tilde{V}^+ \backslash \tilde{Z}$.
\begin{align}
\pi&_\Theta(\theta) = \nonumber \\
&\begin{cases}
    1 & \text{if }\theta=\theta_0, \\
    \varsigma(Pa_\Theta(\theta), \theta)\pi_\Theta(Pa_\Theta(\theta)) & \text{if }\theta \in \Theta^+.
\end{cases} \label{eq:reach-theta} \\
\pi&(\sigma, (\tilde{v}, \theta)) = \pi_\Theta(\theta)\prod_{i\in N_+}\pi_i(\sigma, (\tilde{v}, \theta)) \label{eq:reach}.
\end{align}

The expected payoff for rational player $i \in N_+$, given a strategy profile $\sigma \in \Sigma$, is $\hat{u}_i(\sigma) = \sum_{z \in Z} \pi(\sigma, z) u_i(z)$.
A strategy profile $\sigma^* \in \Sigma$ is a \textbf{Nash equilibrium} if, for all $i \in N_+$, the following holds:
$$
\hat{u}_i(\sigma^*) \geq \max_{\sigma_i^\prime \in \Sigma_i} \hat{u}_i(\sigma^*_{-i} \oplus \sigma_i^\prime).
$$

\subsection{Signal Observation Abstraction}\label{subsec:signal-observation-abstraction}

Applying SOOGs to HUNL, the set of all cards dealt from the beginning of a game up to the current state constitutes a signal. For instance, in a game reaching the Turn phase, player 1 and player 2 are dealt \texttt{A$\heartsuit$K$\heartsuit$} and \texttt{J$\heartsuit$Q$\heartsuit$}, respectively, during the Preflop. The community cards revealed during the Preflop and Turn are \texttt{1$\heartsuit$2$\heartsuit$3$\heartsuit$} and \texttt{4$\heartsuit$}, respectively. The current signal can thus be represented as $\theta=$ \texttt{[(A$\heartsuit$K$\heartsuit$),(J$\heartsuit$Q$\heartsuit$)$\mid$ 1$\heartsuit$2$\heartsuit$3$\heartsuit \mid$ 4$\heartsuit$]}. The observation of $\theta$ by the spectator player ($sp$) is \texttt{[$\cdot \mid$ 1$\heartsuit$2$\heartsuit$3$\heartsuit \mid$ 4$\heartsuit$]}, reflecting the information visible from the audience perspective. More importantly, the observations of $\theta$ by player 1 and player 2 are \texttt{[A$\heartsuit$K$\heartsuit\mid$ 1$\heartsuit$2$\heartsuit$3$\heartsuit \mid$ 4$\heartsuit$]} and \texttt{[J$\heartsuit$Q$\heartsuit\mid$ 1$\heartsuit$2$\heartsuit$3$\heartsuit \mid$ 4$\heartsuit$]}, respectively, which correspond to the hands of player 1 and player 2. This transformation further enables the task of grouping hands in hold'em games, known as hand abstraction, to be reframed as a signal observation abstraction model within the SOOGs framework.

\begin{definition}
In SOOGs, $\alpha = (\alpha_1, \dots, \alpha_n)$ is a signal observation abstraction profile, and $\Psi^{\alpha}_i$ denotes the set of abstracted signal observation infosets for $i \in N_+$ and $\alpha$. Each $\alpha_i: \Theta \mapsto \Psi^{\alpha}_i$ maps a signal $\theta \in \Theta$ to an abstracted signal observation infoset $\hat{\psi} \in \Psi^{\alpha}_i$. Furthermore, each $\hat{\psi}$ can be further partitioned into finer signal observation infosets within $\Psi_i$.
\end{definition}

A (signal observation) abstracted game $\tilde{\mathcal{G}}^\alpha$ is derived by substituting $\vartheta_i$ with $\alpha_i$ in $\tilde{\mathcal{G}}$. However, signal observation abstraction does not directly alter the game itself but instead modifies the structure of the players' strategy spaces. Specifically, the original strategy space $Q_i = \{(\tilde{h}, \psi, a) \in \tilde{H}_i \times \Psi_i \times A_+ \mid \gamma(\tilde{h}) = d_\Theta(\psi) \wedge a \in \mathcal{A}(\tilde{h})\}$ is transformed into (signal observation) abstracted strategy space $Q_i^\alpha = \{(\tilde{h}, \hat{\psi}, a) \in \tilde{H}_i \times \Psi_i^\alpha \times A_+ \mid \gamma(\tilde{h}) = d_\Theta(\hat{\psi}) \wedge a \in \mathcal{A}(\tilde{h})\}$. The abstracted strategy space reveals the essence of the hand abstraction task: it assumes a consistent partitioning of signals across different public nodes within the same phase. This consistency allows the task of abstracting each individual infoset in IIGs to be simplified into abstracting groups of infosets -- the Cartesian product of $\Psi_i$ and a public node -- in SOOGs. 

The concepts of perfect and imperfect recall, originating from IIGs, describe whether players remember all information encountered throughout the game. These notions can be naturally extended to the SOOGs framework. In a $\tilde{\mathcal{G}}$ of SOOG, a player $i \in \tilde{N}_+$ is said to have signal perfect recall if, for any two signals $\theta'_1, \theta'_2 \in \psi'$, every predecessor $\theta_1$ of $\theta'_1$ corresponds to a predecessor $\theta_2$ of $\theta'_2$ such that $\theta_2 \in \vartheta_i(\theta_1)$ and $\theta_1 \in \vartheta_i(\theta_2)$. When all players in the game $\tilde{\mathcal{G}}$ have signal perfect recall, the game itself is said to have signal perfect recall. If $\tilde{\mathcal{G}}$ has perfect recall, let $\alpha_i$ represent the signal observation abstraction for player $i \in N_+$. The abstraction profile $(\alpha_i, \vartheta_{-i})$ refers to a scenario in which player $i$ employs the signal observation abstraction $\alpha_i$, while the other players do not use any abstraction. If the game $\tilde{\mathcal{G}}^{(\alpha_i, \vartheta_{-i})}$ retains signal perfect recall, then $\alpha_i$ is considered a signal observation abstraction with \textbf{perfect recall}; otherwise, it is considered a signal observation abstraction with \textbf{imperfect recall}.

\section{Resolution Bound} \label{sec:resolution-bound}

In this section, we build on the signal observation abstraction model introduced in Section \ref{subsec:signal-observation-abstraction} to propose a novel metric: the \textbf{resolution bound}, which quantifies a signal observation abstraction algorithm's ability to retain information. Algorithms with higher information retention demonstrate stronger competitiveness under the same number of signal observation infosets and are capable of identifying a broader range of abstracted signal observation infosets. This makes the resolution bound an effective metric for evaluating and comparing the performance of signal observation abstraction algorithms. In this context, hand abstraction algorithms can be viewed as specific instances of signal observation abstraction algorithms. A signal observation abstraction algorithm is a procedure that determines a specific signal observation abstraction based on a given parameter set, and the resolution bound provides a unified framework for assessing its effectiveness.

We begin by defining the refinement relationship between signal observation abstractions. It is often challenging to directly compare the granularity of signal observation infosets identified by two abstractions, as one may offer finer granularity for certain signal observation infosets, while the other achieves greater precision for others. The refinement relationship describes a specific case where one signal observation abstraction algorithm is finer than another across all signal observation infosets. The formal definition is as follows:

\begin{definition}
In a SOOG, $\alpha_i$ and $\beta_i$ are signal observation abstractions for player $i$. The refinement relationship between $\alpha_i$ and $\beta_i$ is defined as follows: If, for any $\hat{\psi} \in \Psi_i^{\beta}$, there exist one or more abstracted signal observation infosets in $\Psi_i^{\alpha}$ such that their union forms a partition of $\hat{\psi}$, then $\alpha_i$ is said to refine $\beta_i$, denoted as $\alpha_i \sqsupseteq \beta_i$. Furthermore, $\alpha_i$ is a common refinement of multiple signal observation abstractions $\alpha_i^{1}, \ldots, \alpha_i^{m}$ for player $i$ such that $\alpha_i \sqsupseteq \alpha_i^{j}$ for all $j \in \{1, \ldots, m\}$.
\end{definition}

We can then extend the concept of a common refinement for a signal observation abstraction algorithm $Alg$. If there exists a signal observation abstraction $\alpha_i$ such that, for any set of parameters, the signal observation abstraction $\alpha_i'$ generated by $Alg$ can be refined by $\alpha_i$, then $\alpha_i$ serves as the common refinement of $Alg$ in the game. The common refinement of a signal observation abstraction algorithm defines the upper bound of its ability to distinguish signal observation infosets, which we also refer to as the resolution bound.

For any given signal observation algorithm, multiple resolution bounds can be identified. However, these bounds may deviate significantly from the algorithm's actual maximum capability to distinguish signal observation infosets. To strictly define an algorithm's capacity for information recognition, a tight resolution bound may need to be established. Nevertheless, determining such a tight bound is challenging. While having a high-quality resolution bound for an algorithm does not necessarily imply that the algorithm itself is excellent, the presence of a low-quality resolution bound can reveal inherent deficiencies in the algorithm.

\section{Potential-Aware Abstraction Algorithms and Their Resolution Bound} \label{sec:PAAs-and-excessive-abstraction}

Potential-aware abstraction algorithms (PAAs) are a class of automated abstraction algorithms that consider the potential of hands to transition into future game states. The foundational concept was introduced by Gilpin et al. in 2007, who proposed a potential-aware automated abstraction algorithm (PAAA)~\citep{gilpin2007better}. This algorithm employs the k-means clustering method to classify hands based on the histogram distances between the probability distributions of current hands transitioning into next-round hands. A key feature of this approach is its ability to differentiate hands with identical winning probabilities but varying long-term potentials. In 2014, Ganzfried et al. refined this approach by replacing the $L_2$ distance with the earth mover's distance (EMD)~\citep{rubner2000earth}, resulting in the potential-aware abstraction algorithm with earth mover's distance (PAAEMD)~\citep{ganzfried2014potential}. PAAEMD has since been successfully applied in leading hold'em AI such as DeepStack, Libratus, and Pluribus, establishing itself as the current SOTA hand abstraction algorithm.

In this section, we use the signal observation abstraction model to remodel and construct a lower resolution bound for these algorithms. We highlight that PAAs suffer from significant information loss.

\subsection{Potential-Aware Abstraction Algorithms}

PAAs are designed specifically for two-player scenarios. Without loss of generality, we describe the algorithm from the perspective of a single player $i$ against his opponent.

For a SOOG with $\Gamma$ phases, both the algorithms first determine the number of clusters for each phase, denoted as $(m^{(1)}, m^{(2)}, \dots, m^{(\Gamma)})$, which guides the resolution granularity of abstraction at each phase.

Next, for the $\Gamma$-th phase, the strength of each signal observation $\psi \in \Psi_i^{(\Gamma)}$ is quantified by a value called equity, calculated as:
\begin{equation}
\text{equity}(\psi) = w(\psi) + \frac{1}{2}t(\psi), \label{eq:equity}
\end{equation}
where $w(\psi)$, $t(\psi)$, and $l(\psi) = 1-w(\psi) - t(\psi)$ represent the probabilities that the signals rank player $i$ higher than, tied with, and lower than his opponent over $\psi$, respectively. The \(k\)-means algorithm is then applied to \(\Psi_i^{(\Gamma)}\), dividing it into \(m^{(\Gamma)}\) clusters $c^{(\Gamma)}_1, c^{(\Gamma)}_2, \dots, c^{(\Gamma)}_{m^{(\Gamma)}}$ based on the equity squared differences.

For each earlier phase \( r = 1, \dots, \Gamma-1 \), the algorithm constructs a histogram of length \( m^{(r+1)} \), \( h(\psi) \), for each signal observation \(\psi \in \Psi_i^{(r)}\). The \( j \)-th element of the histogram represents the conditional probability that a signal in the observation infoset \(\psi\) transitions to a signal in phase \( r+1 \) whose observation belongs to cluster \( c^{(r)}_j \). This is computed using \eqref{eq:paa}. The \( k \)-means algorithm is then applied to \(\Psi_i^{(r)}\), clustering it into \( m^{(r)} \) groups \( c^{(r)}_1, c^{(r)}_2, \dots, c^{(r)}_{m^{(r)}} \) based on histogram differences. PAAA utilizes the \( L_2 \) distance, while PAAEMD employs EMD between histograms.

\begin{align}
h_j(\psi) &= P(c^{(r+1)}_j \mid \psi) \nonumber \\
&= \frac{\sum_{\theta \in \psi, \theta' \in c^{(r+1)}_j} \pi_{\Theta}(\theta) \varsigma(\theta, \theta')}{\sum_{\theta \in \psi} \pi_{\Theta}(\theta)}.  \label{eq:paa}
\end{align}



\subsection{Potential-Aware Outcome Isomorphism}

We emulate the process of constructing abstractions in PAAs and construct a resolution bound for them, termed potential-aware outcome isomorphism. The core idea of this signal observation abstraction is to adopt a greedy approach at each step, aiming to identify the maximum possible abstracted signal observation infosets. Subsequently, we rigorously prove that the potential-aware outcome isomorphism serves as the resolution bound of potential-aware abstraction algorithms.

Without loss of generality, we construct the signal observation abstraction from the perspective of player \( i \). First, for any signal observation infoset $\psi \in \Psi_i^{(r)}$ , we define a winrate outcome feature as
\begin{align}
&wo_i^{(r)}(\psi) \nonumber \\
= &\big( wo_i^{(r),0}(\psi), wo_i^{(r),1}(\psi), \ldots, wo_i^{(r),n}(\psi) \big), \label{eq:wof}
\end{align}
where
\begin{itemize}[left=0cm]
    \item \( wo_i^{(r),0}(\psi) \) denotes the probability that player \( i \) ranks lower than at least one other player in the terminal signals, after passing through \( \psi \).
    \item \( wo_i^{(r),l}(\psi) \), for \( l > 0 \), denotes the probability that player \( i \) ranks no lower than any other player and ranks higher than exactly \( l-1 \) other players in the terminal signals, after passing through \( \psi \).
\end{itemize}

Specifically, in the 2-player scenario, $l(\psi)$, $t(\psi)$ and $w(\psi)$ in potential-aware abstraction algorithms correspond to $wo_i^{(r),0}(\psi)$, $wo_i^{(r),1}(\psi)$ and $wo_i^{(r),2}(\psi)$, respectively.

Then, for each signal observation infoset \( \psi \in \Psi^{(r)}_i \), we construct a potential-aware outcome feature (POF) \( pf_i^{(r)}(\psi) \). Signal observation infosets with the same POF form an equivalence class. We define the potential-aware outcome isomorphism (POI) by assigning a unique identifier to each equivalence class $id(pf_i^{(r)}(\psi))$, ranging from 1 to \( \tilde{m}^{(r)} \), where \( \tilde{m}^{(r)} \) is the number of distinct POFs in phase \( r \). For terminal phase \( r = \Gamma \), let \( pf_i^{(\Gamma)}(\psi) = wo_i^{(\Gamma)}(\psi) \). For earlier phases \( r = 1, \dots, \Gamma-1 \), the POF is a $\tilde{m}^{(r+1)}$-length histogram, as in potential-aware abstraction algorithms. The \( j \)-th element of \( pf_i^{(r)}(\psi) \) represents the conditional probability
\begin{align}
pf_{i, j}^{(r)}(\psi) &= \frac{\sum_{\theta \in \psi, id(pf_i^{(r+1)}(\vartheta_i(\theta'))) = j} \pi_{\Theta}(\theta) \varsigma(\theta, \theta')}{\sum_{\theta \in \psi} \pi_{\Theta}(\theta)}.  \nonumber
\end{align}

\begin{theorem}
The POI serves as a resolution bound of algorithm PAAA and PAAEMD.
\end{theorem}

\begin{proof}
To prove that the POI serves as a resolution bound for the PAAs, we begin by analyzing the iterative structure of the \(k\)-means clustering algorithm used in both algorithms. This is an iterative process, where each iteration $t$ consists of two primary steps.

\paragraph{Cluster Assignment Step:} For each signal observation infoset \(\psi \in \Psi_i^{(r)}\), assign it to the cluster \(c_{j,t}^{(r)}\) whose centroid \(\mu_{j,t}^{(r)}\) minimizes the chosen distance metric, formally defined as:

\begin{equation}
    c_{j,t}^{(r)} \ni \psi  \quad \text{if} \quad j = \arg \min_{k} D(\nu(\psi), \mu_{j,t}^{(r)}). \nonumber
\end{equation} 

In the terminal phase of both PAAA and PAAEMD, \(D\) is the squared difference operator, and \(\nu(\psi) = \text{equity}(\psi)\). In the non-terminal phase, \(\nu(\psi) = h(\psi)\), while \(D\) is the \(L_2\) distance operator for PAAA and the EMD operator for PAAEMD.

\paragraph{Centroid Update Step:} After the cluster assignments, update the centroid \(\mu_{j,t+1}^{(r)}\) of each cluster to minimize the total distance between the data points and their cluster centroid:

\begin{equation}
    \mu_{j,t+1}^{(r)} = \arg \min_{\mu} \sum_{\psi \in c_{j,t}^{(r)}} D(\nu(\psi), \mu). \nonumber
\end{equation}

It can be observed that for \(\psi, \psi' \in \Psi_i^{(r)}\), if at every iteration \(t\) we have \(\nu(\psi) = \nu(\psi')\), then \(\psi\) and \(\psi'\) will always be assigned to the same cluster in iteration \(t\). Therefore, to prove the conclusion, it suffices to demonstrate that for signal observation infosets \(\psi, \psi' \in \Psi_i^{(r)}\) with the same POF, they will always satisfy \(\nu(\psi) = \nu(\psi')\). We use mathematical induction to prove the statement. 

In the terminal phase, \(pf_i^{(\Gamma)}(\psi) = pf_i^{(\Gamma)}(\psi')\) implies \(w(\psi) = w(\psi')\) and \(t(\psi) = t(\psi')\). By \eqref{eq:equity}, this ensures \(\text{equity}(\psi) = \text{equity}(\psi')\), proving the base case.

For the inductive step, assume that for any \(\psi^1, \psi^2 \in \Psi_i^{(r)}\), if \(pf_{i}^{(r)}(\psi^1) = pf_{i}^{(r)}(\psi^2)\), then \(\nu(\psi^1) = \nu(\psi^2)\) holds. We aim to prove that for any \(\psi, \psi' \in \Psi_i^{(r-1)}\), if \(pf_{i}^{(r)}(\psi) = pf_{i}^{(r)}(\psi')\), then \(\nu(\psi) = \nu(\psi')\) also holds. By the inductive hypothesis, signal observation infosets with the same POF in phase $r$ are assigned to the same cluster. Let the POI indices in cluster \(c^{(r)}_j\) form the set \(pi^{(r)}_j\), then:

\begin{align}
h_j(\psi) &= \frac{\sum_{\theta \in \psi, \theta' \in c^{(r)}_j} \pi_{\Theta}(\theta) \varsigma(\theta, \theta')}{\sum_{\theta \in \psi} \pi_{\Theta}(\theta)} \nonumber \\
&= \sum_{k\in pi^{(r)}_j}\frac{\sum_{\theta \in \psi, id(pf_i^{(r)}(\vartheta_i(\theta'))) = k} \pi_{\Theta}(\theta) \varsigma(\theta, \theta')}{\sum_{\theta \in \psi} \pi_{\Theta}(\theta)} \nonumber \\
&= \sum_{k\in pi^{(r)}_j} pf_{i,j}^{(r)}(\psi) = \sum_{k\in pi^{(r)}_j} pf_{i,j}^{(r)}(\psi') \nonumber \\
&= \sum_{k\in pi^{(r)}_j}\frac{\sum_{\theta \in \psi', id(pf_i^{(r)}(\vartheta_i(\theta'))) = k} \pi_{\Theta}(\theta) \varsigma(\theta, \theta')}{\sum_{\theta \in \psi'} \pi_{\Theta}(\theta)} \nonumber \\
&= \frac{\sum_{\theta \in \psi', \theta' \in c^{(r)}_j} \pi_{\Theta}(\theta) \varsigma(\theta, \theta')}{\sum_{\theta \in \psi'} \pi_{\Theta}(\theta)} = h_j(\psi'). \nonumber
\end{align}

Thus, \(h_j(\psi) = h_j(\psi')\) for phase \(r-1\), and by extension, \(\nu(\psi) = \nu(\psi')\) in the non-terminal phase, completing the proof.
\end{proof}

\subsection{The Issue of Excessive Abstraction in POI} \label{sec:excessive-abstraction}

PAAs fall under the category of \textbf{outcome-based abstraction}. This methodology determines abstraction through a bottom-up analysis based on the results of game rollouts. In contrast, the lossless isomorphism algorithm (LI)~\citep{gilpin2007lossless} operates outside the scope of outcome-based abstraction, as it identifies signal similarities by analyzing the structure of signals themselves rather than their outcomes. We have proven that POI is the resolution bound for PAAs and, moreover, we can extend this analysis to demonstrate that POI serves as the resolution bound for other existing outcome-based abstraction algorithms, such as the EHS algorithm. We omit this proof in this paper for brevity. In this subsection, we highlight POI as an abstraction with relatively low resolution to illustrate the pressing need for improvement in current outcome-based abstraction algorithms.

Table \ref{tbl:defect} presents the number of distinct signal observation infosets at each phase of HUNL and HULH, alongside the number of distinct abstracted signal observation infosets identified by both the LI and POI algorithms at each phase. As the game progresses, the number of distinct signal observation infosets gradually increases, forming a triangular pattern. The LI algorithm effectively captures this pattern while abstracting the games. In contrast, POI fails to reflect this trend, with its abstracted signal observation infosets exhibiting a spindle-shaped distribution—fewer in the early and late phases of the game, and more in the middle phases. This indicates that POI suffers from severe information loss during the later phases of the game, an issue we refer to as excessive abstraction.

\begin{table}[bt]
\centering
\resizebox{\columnwidth}{!}{
\begin{tabular}{cccc}
\specialrule{1.2pt}{0pt}{0pt}
\hline
Phase   & Null Abstraction                  & LI         & POI \\ \hline
Preflop & $\binom{52}{2}$=1326              & 169        & 169                                               \\
Flop    & $\binom{52}{2,3}$=25989600        & 1286792    & 1137132                                                 \\
Turn    & $\binom{52}{2,3,1}$=1195521600    & 55190538   & 2337912                                                \\
River   & $\binom{52}{2,3,1,1}$=5379847200  & 2428287420 & 20687                                                          \\ \hline
\bottomrule[1.2pt]
\end{tabular}
}
\caption{Quantity of signal observation infosets in unabstracted game and abstracted signal observation infosets for various algorithms in HUNL and HULH.}
\label{tbl:defect}
\end{table}

Let's explore the reasons behind this phenomenon. PAAs are considered imperfect-recall abstraction algorithms that allow players to relax the requirement for complete memory retention. However, their construction, including POIs that mimic this process, represents an extreme scenario where no historical information is retained, and decisions are based solely on future signals. This excessive application of imperfect recall, which disregards historical consistency entirely, is referred to as the future-consider-only approach. As shown in Figure \ref{fig:potential_outcome_isomorphism_information}, the scale of abstracted signal observation infosets is influenced by two main factors under the future-consider-only approach. First, the quantity of signal observation infosets increases as the game progresses, enabling the identification of more distinct isomorphism classes. Second, the volume of information within the potential outcome features plays a critical role. For example, in the River phase, potential outcome features account for all possible showdown outcomes considering the opponent's hands when two private cards are drawn from the remaining 45 cards, \(\binom{52-7}{2} = 990\). At the Preflop phase, these features include scenarios where, after removing one’s two private cards from the deck, all possible combinations of the opponent’s private cards and five community cards are considered, \(\binom{52-2}{2} \cdot \binom{50-2}{3} \cdot \binom{48-3}{1} \cdot \binom{45-1}{1} = 41,951,448,000\). Imagine if players retained all past information; in the River phase, they would also retain the memory from the Preflop phase, resulting in an information volume of at least \(41,951,448,000\). However, future-consider-only approaches, such as POI, reduce this to just \(990\). This drastic reduction in information sharply increases the likelihood of repetition in potential outcome features, thereby shrinking the scale of distinct abstracted signal observation infosets. Together, these factors create a spindle-shaped distribution in the quantity of distinct POI classes, with excessive abstraction becoming more pronounced in later phases.

\begin{figure}[bt]
\centering
\includegraphics[width=0.5\textwidth]{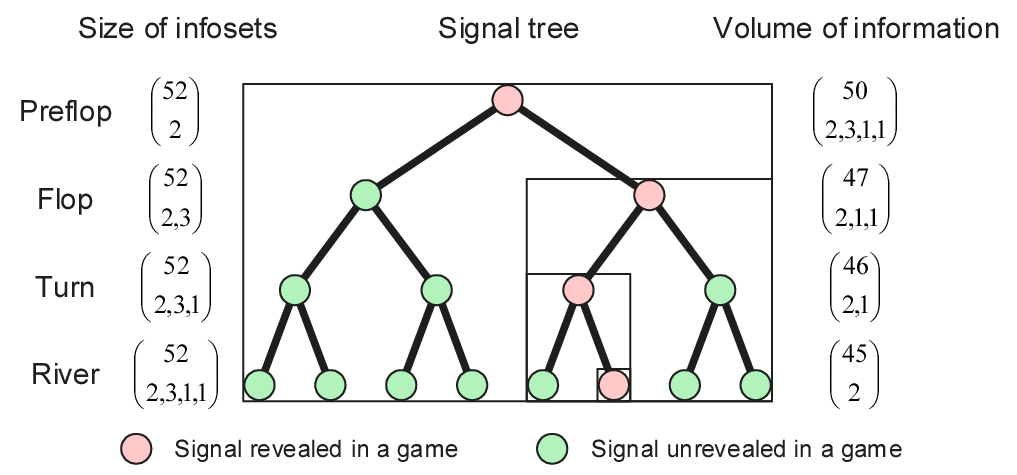}
\caption{Two factors influencing the quantity of distinct abstracted signal observation infosets for future-consider-only approach.}
\label{fig:potential_outcome_isomorphism_information}
\end{figure}

\begin{table*}[ht]
\centering
\resizebox{\textwidth}{!}{%
\begin{tabular}{ccccccccccc}
\specialrule{1.2pt}{0pt}{0pt}
\hline
         & Preflop & \multicolumn{2}{c}{Flop} & \multicolumn{3}{c}{Turn}     & \multicolumn{4}{c}{River}                \\
Phase    & 1 & \multicolumn{2}{c}{2} & \multicolumn{3}{c}{3}     & \multicolumn{4}{c}{4} \\
         \cmidrule(lr){2-2} \cmidrule(lr){3-4} \cmidrule(lr){5-7} \cmidrule(lr){8-11}
Recall   & 0       & 0           & 1           & 0       & 1        & 2        & 0     & 1        & 2         & 3         \\ 
KRWI     & 169     & 1028325     & 1123442     & 1850624 & 34845952 & 37659309 & 20687 & 33117469 & 529890863 & 577366243 \\ \hline \bottomrule[1.2pt]
\end{tabular}%
}
\caption{The number of signal observation infoset equivalence classes identified by KRWI in each phase and $k$ of HUNL and HULH.}
\label{tbl:krwil}
\end{table*}

\section{K-Recall Methods}\label{sec:k-recall-methods}

In Section \ref{sec:excessive-abstraction}, we analyzed the cause of excessive abstraction and identified that it stems from insufficient utilization of historical information. This leads to a straightforward intuition: better leveraging historical information can effectively address this issue. In this section, we construct a signal observation abstraction that incorporates historical information (Section \ref{sec:k-recall-winrate-isomorphism}) and then develop a signal observation abstraction algorithm based on it (Section \ref{sec:krwemd}).

\subsection{K-Recall Winrate Isomorphism} \label{sec:k-recall-winrate-isomorphism}

There are various ways to incorporate historical information. Here, we propose a specific approach: retaining information from the current phase and up to the preceding \(k\) phases. We use the term \(k\)-recall to denote this idea.

Next, we consider what type of information to retain. In PAAs, the histogram of signal observation's abstracted classification for the next phase is used as a feature. However, this information is challenging to compare: the \(L_2\) distance fails to capture differences between classifications, while the earth mover's distance is computationally expensive, difficult to weight appropriately, and its approximations may result in non-symmetric distances. Therefore, we choose a more easily comparable type of data—winrate. Specifically, for player \(i\) at phase \(r\), the signal observation infoset \(\psi \in \Psi_i^{(r)}\) has a \(k\)-recall winrate feature (KRWF)—where \(k < r\)—represented as a numerical array with a dimension of \((k+1)\cdot(n+1)\):
\begin{align}
&rf_i^{(r,k)}(\psi)  \nonumber \\
=&\big(wo_i^{(r)}(\psi); wo_i^{(r-1)}(\psi); \dots; wo_i^{(r-k)}(\psi)\big), \nonumber
\end{align}
where \(wo_i^{(r)}(\psi)\) is computed based on \eqref{eq:wof}. By adjusting \(k\), this approach controls the extent to which historical winrate data is considered. When \(k=0\), only the current phase's winrate data is used, while \(k=r-1\) incorporates winrate data from all previous phases. Further, we define a \(k\)-recall winrate isomorphism (KRWI) as an equivalence class where signal observation infosets with the same KRWF belong to the same KRWI class.

Table \ref{tbl:krwil} presents the number of distinct \(k\)-recall winrate isomorphism classes for different \(k\) values across phases in HUNL and HULH. It can be observed that for \(k=0\), the equivalence classes of KRWI are slightly fewer than those of POI, whereas for \(k=r-1\), it is significantly higher than POI. Moreover, the number of identified abstracted signal observation infosets increases progressively across phases, following a triangular pattern, indicating that the method avoids excessive abstraction. Note that \(k=r-1\) does not imply perfect recall but rather utilizes the winrate data from all preceding phases.

\subsection{KrwEmd Algorithm} \label{sec:krwemd}

In this section, we develop a signal observation abstraction algorithm based on KRWF, enabling flexible adjustment of the number of abstracted signal observation infosets. When constructing KRWF, we specifically considered the need to compare distances, which allows us to further classify KRWF using k-means. For the signal observation infosets $\psi$ and $\psi'$ of player $i$ at phase $r$, we can define the distance of their k-recall winrate feature as
\begin{align}
&D(rf_i^{(r,k)}(\psi), rf_i^{(r,k)}(\psi')) \nonumber \\
= &\sum_{j=0}^k w_j \cdot \text{Emd}(pf_i^{(r-j)}(\psi), pf_i^{(r-j)}(\psi')). \label{equa:distance-of-krwemd}
\end{align}

Among \eqref{equa:distance-of-krwemd}, \text{Emd} is the operator used to calculate the earth mover's distance. The earth mover's distance can be formulated as a linear programming problem. Given two distributions $\vp = (p_1, p_2, \dots, p_n)$ and $\vq = (q_1, q_2, \dots, q_m)$ over two sets of points, and a distane matrix $\mathbf{D} = [d_{ij}]_{n\times m}$ representing the ground distances between each point in $\vp$ and $\vq$, the goal is to find the optimal flow $\mathbf{F} = [f_{ij}]_{n\times m}$ that minimizes the total transportation cost $$\text{Emd}(\vp, \vq) = \min \sum_{i=1}^{n} \sum_{j=1}^{m} w_{ij} d_{ij}$$ subject to the following constraints:
\begin{align}
    &\displaystyle\sum_{j=1}^{m} f_{ij} = p_i, \quad \forall i = 1, 2, \dots, n \nonumber \\ 
    &\displaystyle\sum_{i=1}^{n} f_{ij} = q_j, \quad \forall j = 1, 2, \dots, m \nonumber \\ 
    &f_{ij} \geq 0, \quad \forall i, j \nonumber 
\end{align}
where $f_{ij}$ represents the amount of flow from $p_i$ to $q_j$. Since it requires solving linear programming equations, the computational complexity of the EMD is sensitive to the dimensionality of the histograms, and approximate algorithms are usually used for larger-scale problems. However, the dimensionality of winrate-based features is small, with a dimension of 3 in a two-player scenario, so we attempt to use a fast algorithm for accurately computing the EMD \citep{BPPH11}. $w_0, \dots, w_k$ are hyperparameters used to control the importance of EMD at each phase $r, \dots, r-k$, and the idea behind this design is to transform the similarity between two infoset equivalence classes into a linear combination of the EMD distances between their k-recall winrate features' winrates across different phases. We use the k-means++ algorithm \citep{arthur2007k} to further cluster the abstracted signal observation infoset equivalence classes of KRWI. We named this algorithm KrwEmd.

\section{Experimental Setup} \label{sec:experimental-setup}


\begin{table}[bt]
  \centering
  \resizebox{\columnwidth}{!}{%
\begin{tabular}{ccccccc}
\specialrule{1.2pt}{0pt}{0pt}
\hline
         & Phase 1 & \multicolumn{2}{c}{Phase 2} & \multicolumn{3}{c}{Phase 3}                   \\
Null Abstraction & 780     & \multicolumn{2}{c}{29640} & \multicolumn{3}{c}{1096680}       \\
LI       & 100     & \multicolumn{2}{c}{2260} & \multicolumn{3}{c}{62020}       \\
POI      & 100     & \multicolumn{2}{c}{2250} & \multicolumn{3}{c}{3957}       \\
         \cmidrule(lr){2-2} \cmidrule(lr){3-4} \cmidrule(lr){5-7} 
Recall   & 0       & 0           & 1           & 0       & 1        & 2               \\ 
KRWI     & 100     & 2234        & 2248     & 3957 & 51000 & 51070  \\ \hline   \bottomrule[1.2pt]
\end{tabular}%
}
\caption{Quantity of signal observation infosets in unabstracted game and abstracted signal observation infosets for various algorithms in Numeral211 hold'em.}
\label{tbl:krwi-kroi-numeral211}
\end{table}

We conducted experiments on the Numeral211 hold'em testbed. Table \ref{tbl:krwi-kroi-numeral211} shows the number of signal observation infosets in Numeral211 hold'em and the number of abstracted signal observation infosets identified by different algorithms.

Let $\alpha = (\alpha_1, \alpha_2)$ represent the signal observation abstraction to be assessed. Its strength is evaluated by measuring the \textbf{exploitability} of the approximate equilibrium derived using the CSMCCFR algorithm \citep{zinkevich2007regret, lanctot2009monte}, considering both symmetric and asymmetric abstraction scenarios.

In two-player games with ordered signals, exploitability measures the extent to which a player's strategy deviates from a Nash equilibrium. For a given strategy profile $\sigma = (\sigma_1, \sigma_2)$, the exploitability $\epsilon(\sigma)$ is computed as the difference between the game's expected total payoff at a Nash equilibrium $\sigma^*$ and the expected total payoff of the strategy being played against its best response. Formally, this is defined as
\begin{align}
\epsilon(\sigma) &= \frac{1}{2}(\max_{\sigma'_1 \in \Sigma_1}\hat{u}_1(\sigma'_1 \oplus \sigma_2)-\hat{u}_1(\sigma^*)) \nonumber \\
                 &+ \frac{1}{2}(\max_{\sigma'_2 \in \Sigma_2}\hat{u}_2(\sigma_1 \oplus \sigma'_2)-\hat{u}_2(\sigma^*)), \nonumber
\end{align}
which is measured in terms of milli blinds (antes) per game (mb/g) in Numeral211 hold'em.

In this \textbf{symmetric abstraction scenario}, commonly used in practical applications of abstraction, high-level AIs such as Libratus and DeepStack employ self-play to derive advanced strategies. We measure the exploitability of the approximate equilibrium yielded when both players in the game use signal observation abstraction. However, it may lead to the abstraction pathology \citep{waugh2009abstraction}. To avoid such problems, we illustrate the theoretical performance of the signal observation abstraction under evaluation through \textbf{asymmetric abstraction scenario}. In this scenario, the approximate equilibria of the signal-abstracted games $\tilde{\mathcal{G}}^{(\alpha_1, \vartheta_2)}$ and $\tilde{\mathcal{G}}^{(\vartheta_1, \alpha_2)}$ are computed, yielding strategies $\sigma^{*,1}$ and $\sigma^{*,2}$, respectively. These two strategies are then concatenated to form $\sigma' = (\sigma^{*,1}_1, \sigma^{*,2}_2)$, and the exploitability of $\sigma'$ is evaluated to assess the quality of the abstraction.


Regarding KrwEmd, we set the distance matrix:
\[
\mathbf{D} = \begin{bmatrix} 
0 & 1 & 2 \\ 
1 & 0 & 1 \\ 
2 & 1 & 0 
\end{bmatrix}
\]
For a two-player game, its meaning is quite clear. Taking the first row as an example: transitioning from a loss to a loss costs 0, transitioning to a draw costs 1, and transitioning to a win costs 2.

\section{Experiment} \label{sec:experiment}

\begin{figure}[h!]
    \centering
    \begin{subfigure}[t]{\linewidth}
        \centering
        \includegraphics[width=\linewidth]{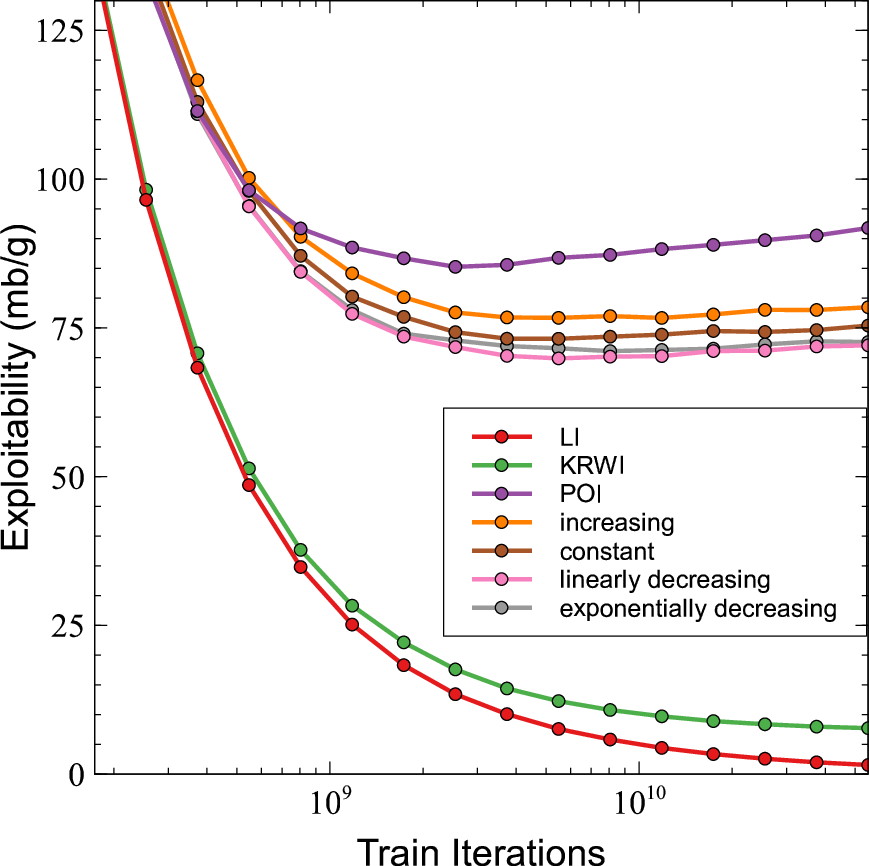}
        \caption{Symmetric abstraction scenario.}
        \label{fig:subfig1}
    \end{subfigure}
    
    \vspace{0cm}  

    \begin{subfigure}[t]{\linewidth}
        \centering
        \includegraphics[width=\linewidth]{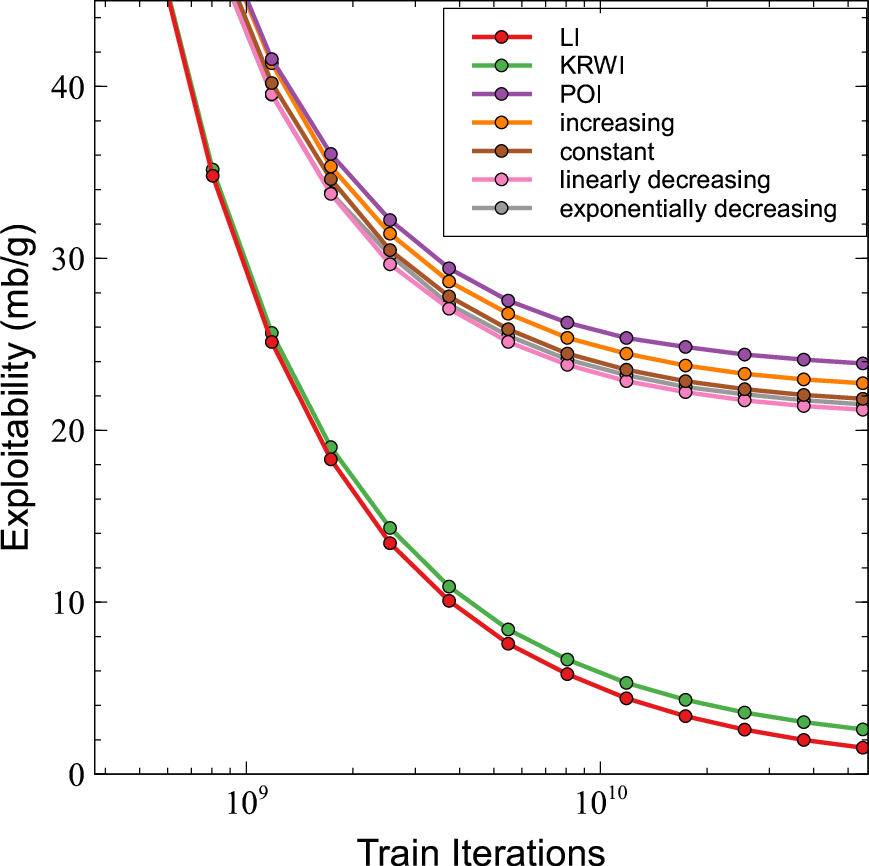}
        \caption{Asymmetric abstraction scenario.}
        \label{fig:subfig2}
    \end{subfigure}

    \caption{Performance comparison of KRWI, POI, and LI, along with ablation study.}
    \label{fig:isomorphism-compare}
\end{figure}


Firstly, we assess the performance of KRWI (with $k = r-1$ for phase $r$), POI, and LI. Note that POI represents the resolution bound of existing future-consider-only signal observation abstraction algorithms, while KRWI serves as the resolution bound for the KrwEmd algorithm. This experiment compares the performance of future-consider-only signal observation abstraction algorithms and KrwEmd at their respective performance ceilings. The experiment involved $5.5 \times 10^{10}$ iterations. Figure \ref{fig:isomorphism-compare} presents the results: (a) illustrates the outcomes under symmetric abstraction scenarios, and (b) shows those under asymmetric abstraction scenarios. Despite the overfitting observed in all algorithms except for LI and KRWI across these settings, the relative performance rankings of the algorithms remained consistent. Specifically, if one abstraction outperformed another in the symmetric setting, the same held true in the asymmetric setting. This consistent behavior across settings indicates the absence of abstraction pathology. Notably, KRWI outperformed POI and demonstrated performance closer to that of lossless abstraction. This highlights that algorithms incorporating historical information achieve a higher performance ceiling compared to those that disregard it. Surprisingly, KRWI's performance was remarkably close to that of LI, despite retaining only 3,957 abstracted signal observation infosets in the 3rd phase of Numeral211 hold'em environment, compared to LI's 62,020. This indicates that by leveraging historical win-rate information, KRWI effectively captures most factors influencing competitiveness, achieving a high cost-effectiveness in information utilization.

Although this experiment underscores the superiority of k-recall algorithms like KrwEmd over future-consider-only algorithms at the performance ceiling, the comparison is somewhat unfair due to the significantly larger scale of abstracted signal observation infosets identified by KRWI compared to POI. To address this, we conducted an ablation study to further validate the advantages of incorporating historical information. In the ablation experiment, we adjusted KrwEmd to match the number of abstracted signal observation infosets identified by POI for a fair comparison within the isomorphism framework. Notably, POI exhibit similar capabilities in recognizing signal infoset equivalence classes as KRWI ($k=0$ and $k=1$) in phases 1 and 2. as shown in Table \ref{tbl:krwi-kroi-numeral211}. Therefore, we directly used the equivalence classes identified by POI in phases 1 and 2, performing KrwEmd clustering ($k=2$) only in phase 3. We designed four sets of hyper-parameters $(w_0, w_1, w_2)$ in equation \eqref{equa:distance-of-krwemd} to adjust the importance of historical information, namely exponentially decreasing $(16, 4, 1)$, linearly decreasing $(7, 5, 3)$, constant $(1, 1, 1)$, and increasing $(3, 5, 7)$. However, excessively increasing the importance of later phases did not lead to sustained improvement (e.g., exponentially decreasing performed worse than linearly decreasing). In fact, if the weight of later phases approaches an extreme, the algorithm would eventually degrade into a future-consider-only algorithm.

\begin{figure}[h!]
    \centering
    \begin{subfigure}[t]{\linewidth}
        \centering
        \includegraphics[width=\linewidth]{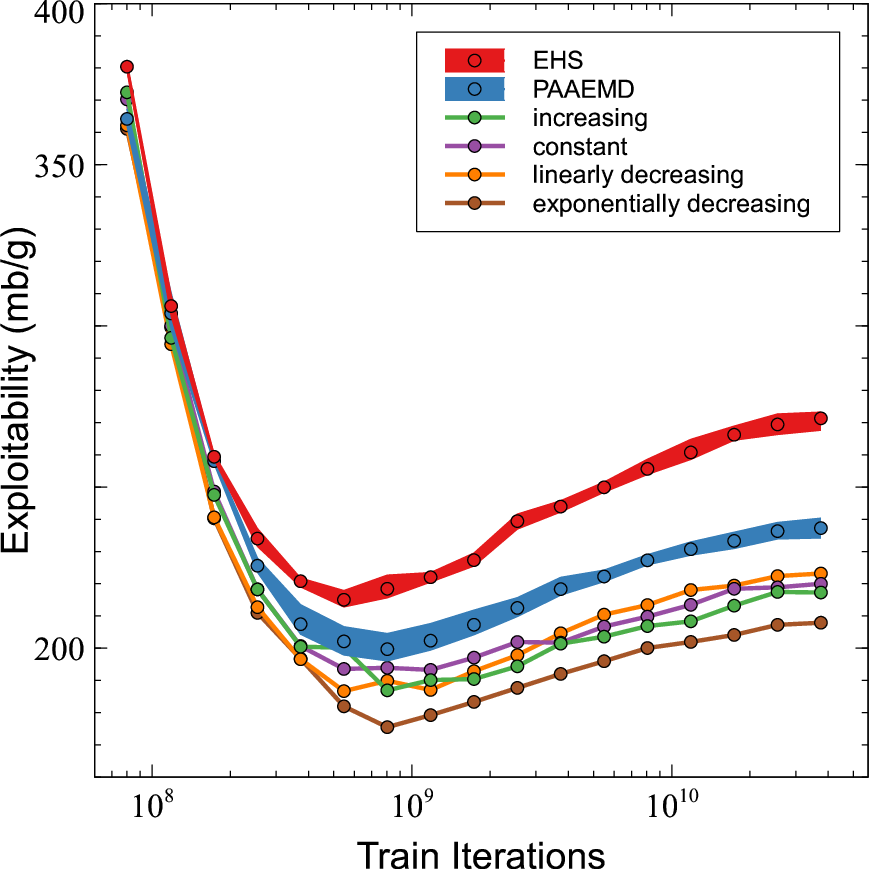}
        \caption{Symmetric abstraction scenario.}
        \label{fig:abstraction-symmetric}
    \end{subfigure}
    
    \vspace{0cm}  

    \begin{subfigure}[t]{\linewidth}
        \centering
        \includegraphics[width=\linewidth]{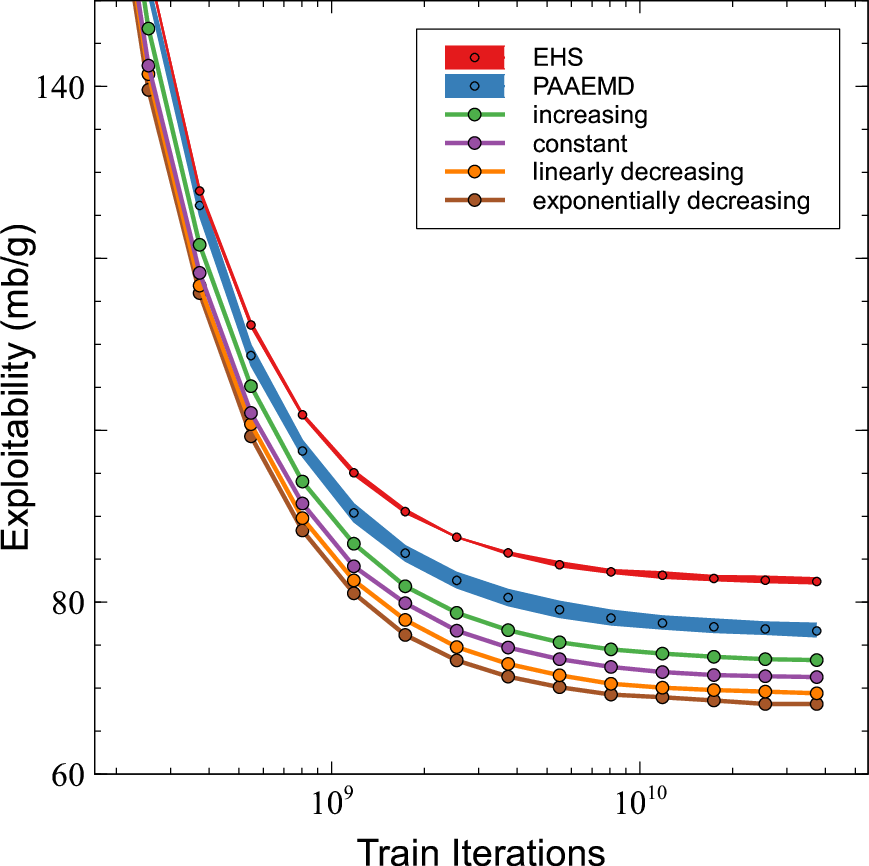}
        \caption{Asymmetric abstraction scenario.}
        \label{fig:abstraction-asymmetric}
    \end{subfigure}

    \caption{Performance comparison of PAAEMD, EHS, and KrwEmd with different parameters.}
    \label{fig:abstraction-compare}
\end{figure}

Next, we compared KrwEmd's performance with two widely used signal observation abstraction algorithms, EHS and PAAEMD. Notably, POI serves as the resolution bound for both EHS and PAAEMD, meaning the maximum number of signal observation infoset equivalence classes they can recognize does not exceed that of POI. For the experiment, we set a compression rate 10 times lower than that of POI and performed no abstraction for phase 1. The final numbers of abstracted signal observation infosets were set to 100, 225, and 396 for phases 1, 2, and 3, respectively. To eliminate the influence of randomness on performance, we generated three sets of abstractions for both EHS and PAAEMD. For KrwEmd, we used hyperparameters $(w_{3,0}, w_{3,1}, w_{3,2}; w_{2,0}, w_{2,1})$ in phases 3 and 2, with four configurations of historical importance: exponentially decreasing $(16, 4, 1; 4, 1)$, linearly decreasing $(7, 5, 3; 5, 3)$, constant $(1, 1, 1; 1, 1)$, and increasing $(3, 5, 7; 5, 7)$.

Figure \ref{fig:abstraction-compare} presents the results. Additionally, Figure \ref{fig:abstraction-symmetric} shows the outcomes for the symmetric abstraction scenario, while Figure \ref{fig:abstraction-asymmetric} illustrates the results for the asymmetric abstraction scenario. We observed that both symmetric and asymmetric abstractions maintained consistent performance, similar to the previous experiment, without significant abstraction pathologies. However, noticeable overfitting occurred in all abstraction algorithms under the symmetric abstraction scenario. The results demonstrate that KrwEmd significantly outperforms both EHS and PAAEMD across all parameter configurations. Moreover, we confirmed that the importance of historical information decreases progressively from the late game to the early game, although this time, the best-performing parameter configuration decreased exponentially rather than linearly.

By providing a fair comparison, these two experiments validate that considering historical information is indeed more effective than the future-consider-only approach in signal observation abstraction.

\section{Conclusion} \label{sec:conclusion}


This research focuses on the hand abstraction task in the construction of Hold'em AI. We specialize a subset of imperfect information games—signal observation sequential games—to model Hold'em games, enabling the use of signal observation abstraction models for hand abstraction tasks. Based on this model, we propose a reasonable metric for evaluating signal observation abstraction algorithms: the resolution bound. This research introduces the first imperfect recall signal observation abstraction algorithm that considers historical information.  We demonstrate the limitations of existing future-consider-only algorithms by constructing a low resolution bound and propose an improved approach, the KrwEmd algorithm, which incorporates historical information. Experiments show that KrwEmd outperforms all existing algorithms. Since hand abstraction is a relatively independent module in the construction of Hold'em AI, and previous state-of-the-art systems such as DeepStack, Pluribus, and Libratus relied on the future-consider-only algorithm PAAEMD, this research has the potential to significantly enhance the performance of high-level AI.

\bibliography{references}
\bibliographystyle{plain}

\end{document}